\documentclass[aps,pra,twocolumn,nofootinbib, superscriptaddress,10pt]{revtex4-1}

\usepackage{graphicx}
\usepackage{epstopdf}
\usepackage{amsmath}
\usepackage{amssymb}
\usepackage{mathrsfs}
\usepackage{amsthm}
\usepackage{bm}
\usepackage{url}
\usepackage[T1]{fontenc}
\usepackage{csquotes}
\MakeOuterQuote{"}

\usepackage{color}

\newtheoremstyle{note}
  {\topsep/2}               
  {\topsep/2}               
  {}                      
  {\parindent}            
  {\itshape}              
  {.}                     
  {5pt plus 1pt minus 1pt}
  {}

\theoremstyle{note}

\newtheorem{proposition}{Proposition}

\theoremstyle{definition}

\theoremstyle{remark}

\newcommand{\tr}{\operatorname{tr}}

 \newcommand{\rmd}{\mathrm{d}}
 \newcommand{\rme}{\mathrm{e}}
  
 \newcommand{\rmi}{\mathrm{i}}

 \newcommand{\rmA}{\mathrm{A}}

 \newcommand{\rmE}{\mathrm{E}}

 \newcommand{\rmM}{\mathrm{M}}
 
 \newcommand{\rmP}{\mathrm{P}}

 \newcommand{\rmU}{\mathrm{U}}

  \newcommand{\caB}{\mathcal{B}}
  \newcommand{\caE}{\mathcal{E}}
 
 \newcommand{\caH}{\mathcal{H}}

\newcommand{\caU}{\mathcal{U}}
\newcommand{\caV}{\mathcal{V}}

 \newcommand{\bbZ}{\mathbb{Z}}

 \newcommand{\mkB}{\mathfrak{B}}

\newcommand{\be}{\begin{equation}}
\newcommand{\ee}{\end{equation}}
\newcommand{\ba}{\begin{align}}
\newcommand{\ea}{\end{align}}

\def\<{\langle}  
\def\>{\rangle}  

\def\eqref#1{\textup{(\ref{#1})}}  
\newcommand{\eref}[1]{Eq.~\textup{(\ref{#1})}}
\newcommand{\Eref}[1]{Equation~\textup{(\ref{#1})}}
\newcommand{\esref}[1]{Eqs.~\textup{(\ref{#1})}}

\newcommand{\tref}[1]{Table~\ref{#1}}

\newcommand{\pref}[1]{Proposition~\ref{#1}}

\newcommand{\cref}[1]{Conjecture~\ref{#1}}
\newcommand{\Cref}[1]{Conjecture~\ref{#1}}

\newcommand{\rcite}[1]{Ref.~\cite{#1}}
\newcommand{\rscite}[1]{Refs.~\cite{#1}}

\begin{document}
\title{Efficient verification of  quantum gates with local operations}

\author{Huangjun Zhu}
\email{zhuhuangjun@fudan.edu.cn}

\affiliation{Department of Physics and Center for Field Theory and Particle Physics, Fudan University, Shanghai 200433, China}

\affiliation{State Key Laboratory of Surface Physics, Fudan University, Shanghai 200433, China}

\affiliation{Institute for Nanoelectronic Devices and Quantum Computing, Fudan University, Shanghai 200433, China}

\affiliation{Collaborative Innovation Center of Advanced Microstructures, Nanjing 210093, China}

\author{Haoyu Zhang}
\affiliation{Department of Physics and Center for Field Theory and Particle Physics, Fudan University, Shanghai 200433, China}

\affiliation{State Key Laboratory of Surface Physics, Fudan University, Shanghai 200433, China}

\begin{abstract}
Efficient verification of the functioning of quantum devices is a key to the development of quantum technologies, but is a daunting task as the system size increases. Here we propose  a simple and general framework for verifying unitary transformations that can be applied to both individual quantum gates and gate sets, including quantum circuits. This framework enables efficient verification of  many important unitary transformations, including but not limited to  all bipartite unitaries, Clifford unitaries, generalized controlled-$Z$ gates,  generalized controlled-NOT gates,  the  controlled-SWAP gate, and permutation transformations. For all these unitaries, the sample complexity increases at most linearly with the system size  and is often independent of the  system size.
Moreover, little overhead is incurred even if one can only prepare Pauli eigenstates and perform local  measurements.
Our approach  is applicable in many scenarios in which randomized benchmarking (RB) does not apply and is thus instrumental to quantum computation and many other applications in quantum information processing. 
\end{abstract}

\date{\today}
\maketitle

\section{Introduction}      
Quantum technologies  promise to dramatically boost our capability in secure communication, fast computation, and efficient simulation of quantum many-body systems.  To harness the power of quantum technologies, it is crucial to characterize and verify quantum devices with a high precision \cite{EiseHWR19}. The problem is particularly pressing in the context of quantum computation, in which efficient characterization of quantum gates and circuits has become a bottleneck. To characterize quantum gates with quantum process tomography, the resource overhead  increases exponentially with the system size. As a popular alternative, 
randomized benchmarking (RB) \cite{EmerAZ05,KnilLRB08,MageGE11,WallF14} is much more efficient and is robust against state preparation and measurement  (SPAM) errors. However, RB usually 
relies on strong assumptions on the set of gates to be benchmarked, such as group structure and the property of being a unitary 2-design. Although tremendous efforts have been directed to relaxing these assumptions \cite{MageGJR12,CariWE15,CrosMBS16,HarpF17,OnorWE19,HelsXVW19}, the applicability of RB is still quite limited.

In this paper we propose a general and efficient framework for verifying 
unitary processes  based on local operations, that is, 
local state preparation and measurements; only Pauli operations are required in many cases of practical interest.
Our approach  can be applied to individual gates as well as gate sets, including quantum circuits, even in various scenarios in which RB does not apply. 
It enables efficient verification of
many important unitary transformations, including but not limited to all (qubit and qudit) bipartite unitaries,  Clifford unitaries, generalized controlled-$Z$ gates,  generalized controlled-NOT gates  (CNOT),  the  controlled-SWAP gate (CSWAP, also known as the  Fredkin gate) \cite{NielC10book}, and permutation transformations. To achieve infidelity $\epsilon$ and significance level $\delta$ in verifying these transformations, the number of tests required (sample complexity) is at most $\lceil 2n\epsilon^{-1}\ln\delta^{-1}\rceil $ ($n$ is the number of parties), and it is even independent of the system size in many cases of practical interest.

\bigskip

\section{Channel-state duality}               Let $\caH$ be a Hilbert space of dimension $d$ and $\mkB(\caH)$ the space of bounded linear operators on $\caH$. 
A quantum channel or process $\Lambda$ on $\mkB(\caH)$ is represented by a completely-positive and trace-preserving (CPTP) map \cite{NielC10book}. Given a unitary transformation $\caU$ associated with the unitary operator $U$ on $\caH$, let $\caU^\dag$ denote the transformation associated with $U^\dag$, the Hermitian conjugate of $U$. The \emph{average gate fidelity} between $\Lambda$ and $\caU$ reads
\begin{align}
F_\rmA(\Lambda, \caU):=&\int \tr[ \Lambda(|\psi\>\<\psi|)U|\psi\>\<\psi|U^\dag]\rmd\psi\nonumber\\
=&\int \<\psi|\caE(|\psi\>\<\psi|)|\psi\>\rmd\psi,
\end{align}
where the integral is taken over the normalized Haar measure and $\caE=\caU^\dag\circ \Lambda$ is the composition of $\Lambda$ and $\caU^\dag$. The  average gate infidelity reads $\epsilon_\rmA:=1-F_\rmA$.

The idea of channel-state duality \cite{Choi75}
will play a key role in  studying quantum gate verification (QGV). Quantum channels on  $\mkB(\caH)$ are in one-to-one correspondence with  Choi states on $\caH\otimes\caH$ \cite{Choi75}. More precisely, the  \emph{Choi state} associated with the channel $\Lambda$ is defined as 
\begin{equation}
\chi_\Lambda:=(\Lambda\otimes 1)(|\Phi\>\<\Phi|),
\end{equation}
where $\Phi=(\sum_j |jj\>)/\sqrt{d}$ is a maximally entangled state in $\caH\otimes \caH$.
The second reduced state of $\chi_\Lambda$ is completely mixed, that is, $\tr_1\chi_\Lambda=1/d$. Conversely, any quantum state on $\caH\otimes \caH$ that  satisfies this condition determines a quantum channel via the duality relation. Denote by $\rho^*$ the complex conjugate of $\rho$ in  the computational basis, then we have
\begin{equation}
\Lambda(\rho)=d\tr_2[\chi_\Lambda (1\otimes \rho^*)].
\end{equation}

The \emph{entanglement fidelity} or process fidelity between $\Lambda$ and $\caU$  is defined as the fidelity between the Choi states $\chi_\Lambda$ and $\chi_\caU$ \cite{Schu96}, that is, 
\begin{equation}
F_\rmE(\Lambda, \caU):=\tr(\chi_\Lambda\chi_\caU)=\<\Phi| \chi_\caE|\Phi\>, 
\end{equation}
where $\caE=\caU^\dag\circ \Lambda$; the entanglement or process infidelity reads $\epsilon_\rmE:=1-F_\rmE$. 
By \rscite{HoroHH99G,Niel02}, we have
\begin{equation}\label{eq:AGFvsEF}
F_\rmA=\frac{d F_\rmE+1}{d+1},\quad  \epsilon_\rmA=\frac{d \epsilon_\rmE}{d+1}. 
\end{equation}
 When $\caU$ is the identity, $F_\rmA(\Lambda, \caU)$ and $F_\rmE(\Lambda, \caU)$
are abbreviated as $F_\rmA(\Lambda)$ and $F_\rmE(\Lambda)$, respectively.

\section{Quantum gate verification}      
\subsection{Basic framework}
Consider a device that is supposed to 
 perform the unitary transformation $\caU$, but actually realizes an unknown  channel $\Lambda$. To verify whether the channel $\Lambda$ is sufficiently close to the target unitary, we can  pick a  pure test state $\rho_j=|\psi_j\>\<\psi_j|$ randomly with probability $p_j$ from an ensemble of test states $\{\rho_j, p_j\}_j$ and feed it into the channel. Here a state ensemble refers to a weighted set (including a multiset) of quantum states in which the weights form a probability distribution. Then we can  verify whether the output state $\Lambda(\rho_j)$  is sufficiently close to the target state $\caU(\rho_j)=U\rho_j U^\dag$. 
To this end, we perform two-outcome tests $\{E_{l|j}, 1-E_{l|j}\}$ from a set of  accessible tests depending on the input state $\rho_j$ \cite{PallLM18,ZhuH19AdS,ZhuH19AdL}. 
The test operator $E_{l|j}$ corresponds to passing the test and satisfies the conditions $0\leq E_{l|j}\leq1$ and  $E_{l|j}\caU(\rho_j)=\caU(\rho_j)$, so that the target output state $\caU(\rho_j)$ can always pass the test.

Suppose the test $E_{l|j}$ is performed with probability $p_{l|j}$ given the test state $\rho_j$, then the passing probability of $\Lambda(\rho_j)$ reads $\tr[\Omega_j\Lambda(\rho_j)]$, where $\Omega_j=\sum_lp_{l|j} E_{l|j}$ is a \emph{verification operator} for $\caU(\rho_j)$ \cite{PallLM18,ZhuH19AdS,ZhuH19AdL}. The overall average passing probability reads
\begin{equation}\label{eq:AvePassProb}
\sum_{j\,l} p_j p_{l|j} \tr[\Lambda(\rho_j)E_{l|j}]=\tr(\tilde{\Theta}\chi_\Lambda)=\tr(\Theta\chi_\caE), 
\end{equation}
where
\begin{align}\label{eq:PVO}
\tilde{\Theta}:=d\sum_{j}p_j \Omega_j\otimes \rho_j^*,\quad
\Theta&:=d\sum_{j}p_j \caU^\dag(\Omega_j)\otimes \rho_j^*
\end{align}
are called \emph{process} or \emph{channel verification operators}. Here the first variant $\tilde{\Theta}$ is more natural, but the second variant $\Theta$ is more convenient for technical analysis. When  $\Lambda=\caU$ (which means $\chi_\caE=|\Phi\>\<\Phi|$), we have
\begin{equation}
\tr(\tilde{\Theta}\chi_\Lambda)=\tr(\Theta\chi_\caE)=\sum_j p_j \tr[\Omega_j\caU(\rho_j)]=\sum_j p_j=1, 
\end{equation}
given that  $\Omega_j\geq \caU(\rho_j)$ and  $\caU(\rho_j)$ is supported in the eigenspace of $\Omega_j$ with the largest eigenvalue~1. So the passing probability is 1  as expected.

When  $\Omega_j= \caU(\rho_j)$, the process verification operator $\Theta$ reduces to the \emph{preparation operator}
\begin{align}
\Theta_\rmP:=d\sum_j p_j\rho_j\otimes \rho_j^*.
\end{align}
This operator encodes the key information about the ensemble of test states and determines the quality of this ensemble for QGV. In general, the spectral gap of $\Omega_j$ is defined as $\nu_j=1-\beta_j$, where  
 $\beta_j:=\|\Omega_j-\caU(\rho_j)\|$
is  the second largest eigenvalue of $\Omega_j$ \cite{PallLM18,ZhuH19AdS,ZhuH19AdL}. The \emph{measurement spectral gap} of the verification strategy is the minimum of $\nu_j$ over all $j$, that is,
\begin{align}\label{eq:nuM}
\nu_\rmM:=\min_j \nu_j=1-\beta_\rmM,
\end{align}
where $\beta_\rmM:=\max_j \beta_j$. As we shall see shortly, to a large extent $\nu_\rmM$ characterizes the quality of  the verification operators of output states. 
Since  $\Omega_j\leq \nu_j \caU(\rho_j)+\beta_j$, we can deduce the following inequality,
\begin{align}\label{eq:ThetaBound}
\Theta&\leq d\sum_j p_j (\nu_j\rho_j+\beta_j)\otimes \rho_j^* \leq \nu_\rmM\Theta_\rmP
+\beta_\rmM\otimes \sum_j d p_j \rho_j^*,
\end{align}
which is useful to studying the efficiency of  QGV.

\subsection{Number of tests required}
Now we repeat the preparation and test procedure $N$ times and accept the device iff all tests are passed. In general, the channels $\Lambda_1, \Lambda_2, \ldots, \Lambda_N$ realized over the $N$ runs may be different due to inevitable fluctuations, but here we assume that they are independent for simplicity. (More general situations can be addressed using the recipe for the adversarial scenario proposed in \rscite{ZhuH19AdS,ZhuH19AdL}.) Then the acceptance probability  reads $\prod_{r=1}^N \tr(\Theta\chi_{\caE_r})$, 
where $\caE_r=\caU^\dag \circ\Lambda_r$. Our goal is to ensure that the false acceptance probability is smaller than a given threshold, the significance level $\delta$, whenever the average gate (or entanglement) fidelity is smaller than a given threshold.

For $0\leq \epsilon\leq 1$, define
\begin{align}
p_\rmA(\Theta,\epsilon)&:=\max_{F_\rmA(\caE)\leq  1-\epsilon} \tr(\Theta\chi_\caE)=\max_{F_\rmA(\Lambda, \caU)\leq  1-\epsilon} \tr(\tilde{\Theta}\chi_\Lambda),\label{eq:MaxPPAF} \\
p_\rmE(\Theta,\epsilon)&:=\max_{F_\rmE(\caE)\leq  1-\epsilon} \tr(\Theta\chi_\caE)=\max_{F_\rmE(\Lambda, \caU)\leq  1-\epsilon} \tr(\tilde{\Theta}\chi_\Lambda),  \label{eq:MaxPPEF}
\end{align} 
where $\tilde{\Theta}=(\caU\otimes 1)(\Theta)$ by \eref{eq:PVO}.
Then the passing probability in \eref{eq:AvePassProb} is upper bounded by $p_\rmA(\Theta,\epsilon)$ [$p_\rmE(\Theta,\epsilon)$]
whenever the average gate fidelity $F_\rmA(\Lambda,\caU)$
[entanglement fidelity $F_\rmE(\Lambda,\caU)$] is upper bounded by $1-\epsilon$. According to \eref{eq:AGFvsEF}, we have
\begin{equation}\label{eq:pApE}
p_\rmA(\Theta,\epsilon)=p_\rmE(\Theta,d^{-1}(d+1)\epsilon),
\end{equation}
so to compute $p_\rmA(\Theta,\epsilon)$, it suffices to compute $p_\rmE(\Theta,\epsilon)$. In general, it is not easy to derive an analytical formula for $p_\rmE(\Theta,\epsilon)$, but it is easy to compute its value via semidefinite programming (SDP),
\begin{equation}\label{eq:pESDP}
\begin{aligned}
&\mathrm{maximize}  \quad \tr(\Theta\chi_\caE) \quad \mbox{subject to} \\
&\chi_\caE\geq0, \quad \tr_1 \chi_\caE=1/d,\quad 
\<\Phi|\chi_\caE|\Phi\>\leq 1-\epsilon. 
\end{aligned}
\end{equation}
Alternatively, $p_\rmE(\Theta,\epsilon)$ can be computed as follows,
\begin{equation}\label{eq:pESDP2}
\begin{aligned}
&\mathrm{maximize}  \quad \tr(\tilde{\Theta}\chi_\Lambda) \quad \mbox{subject to} \\
&\chi_\Lambda\geq0, \quad \tr_1 \chi_\Lambda=1/d,\quad 
\tr(\chi_\Lambda\chi_\caU)\leq 1-\epsilon. 
\end{aligned}
\end{equation}
These results are reminiscent of  the counterpart in quantum state verification \cite{PallLM18,ZhuH19AdS,ZhuH19AdL}.  

As a simple implication of 
 \esref{eq:MaxPPAF}-\eqref{eq:pESDP}, the following proposition is instructive to understanding QGV. 
\begin{proposition}\label{pro:PassProbCC}
	$p_\rmA(\Theta,\epsilon)$  and $p_\rmE(\Theta,\epsilon)$ are concave and nonincreasing in $\epsilon$, but are convex in $\Theta$.
\end{proposition}
Here the concavity of $p_\rmE(\Theta,\epsilon)$  means
\begin{equation}
p_\rmE(\Theta,\epsilon)\geq \mu p_\rmE(\Theta,\epsilon_1)+(1-\mu)p_\rmE(\Theta,\epsilon_2)
\end{equation}  if $\epsilon= \mu \epsilon_1+(1-\mu)\epsilon_2$ and $0\leq \mu\leq 1$. By contrast, the convexity means 
\begin{equation}
p_\rmE(\Theta,\epsilon)\leq \mu p_\rmE(\Theta_1,\epsilon)+(1-\mu) p_\rmE(\Theta_2,\epsilon)
\end{equation}
if $\Theta=\mu \Theta_1 +(1-\mu)\Theta_2$. Similar properties  hold for $p_\rmA(\Theta,\epsilon)$.

 Thanks to \pref{pro:PassProbCC}, the acceptance probability $\prod_{r=1}^N \tr(\Theta\chi_{\caE_r})$ can be upper bounded by
\begin{align}\label{eq:AcceptPUP}
\prod_{r=1}^N p_\rmE(\Theta,\epsilon_r)\leq \biggl[\frac{1}{N}\sum_r p_\rmE(\Theta,\epsilon_r)\biggr]^N\leq p_\rmE(\Theta,\bar{\epsilon})^N,
\end{align}
where $\epsilon_r=1-F_\rmE(\caE_r)$ and  $\bar{\epsilon}=\sum\epsilon_r/N$ are the entanglement  infidelities over the $N$ runs and their average. The upper bound in \eref{eq:AcceptPUP} can be saturated when all $\epsilon_r$ are equal.  This equation still holds if $p_\rmE$ is replaced by $p_\rmA$.  Define $N_\rmE(\epsilon,\delta,\Theta)$ as the 
minimum number of tests required  to verify  $\caU$ within entanglement infidelity $\epsilon$ and significance level $\delta$, which means $p_\rmE(\Theta,\epsilon)^N\leq \delta$. Then we have
\begin{equation}\label{eq:NumTest}
N_\rmE(\epsilon,\delta,\Theta)=\biggl\lceil
\frac{\ln\delta}{\ln p_\rmE(\Theta,\epsilon)}\biggr\rceil,
\end{equation}
assuming $p_\rmE(\Theta,\epsilon)<1$.
If  entanglement infidelity is replaced by the average gate infidelity, then  the minimum number  is 
\begin{equation}\label{eq:NumTestNA}
N_\rmA(\epsilon,\delta,\Theta)=\biggl\lceil
\frac{\ln\delta}{\ln p_\rmA(\Theta,\epsilon)}\biggr\rceil=N_\rmE(d^{-1}(d+1)\epsilon,\delta,\Theta),
\end{equation}
where the second equality follows from \eref{eq:pApE}.

\subsection{Optimal verification strategies}
To minimize the number of tests required to verify $\caU$, we need to minimize $p_\rmE(\Theta,\epsilon)$ or $p_\rmA(\Theta,\epsilon)$ over the process verification  operator $\Theta$ in \eref{eq:PVO}. 
Given a unitary operator $V$  on $\caH$ and the associated unitary transformation $\caV$, let $\tilde{\Theta}_V$ and $\Theta_V$ be the process verification operators constructed from $\tilde{\Theta}$ and $\Theta$ by replacing $\rho_j$ with $\caV(\rho_j)$ and $E_{l|j}$ with $\caU\caV\caU^\dag(E_{l|j})$,  that is,
\begin{align}
\tilde{\Theta}_V=(\caU\caV\caU^\dag\otimes \caV^*) (\tilde{\Theta}),\quad 
\Theta_V=(\caV\otimes \caV^*) (\Theta). 
\end{align}
Here $\caV^*$ is the transformation associated with $V^*$,  the complex conjugate of $V$ in  the computational basis. Then  $p_\rmE(\Theta_V,\epsilon)=p_\rmE(\Theta,\epsilon)$ and $p_\rmE((\Theta+\Theta_V)/2,\epsilon)\leq p_\rmE(\Theta,\epsilon)$, which implies that averaging over unitarily equivalent strategies cannot decrease the efficiency. 

In addition, the efficiency cannot decrease if we replace $\Omega_j$ with $\caU(\rho_j)$ given that $\Omega_j\geq \caU(\rho_j)$ by assumption. 
Therefore, $p_\rmE(\Theta,\epsilon)$ and  $p_\rmA(\Theta,\epsilon)$ are minimized   when $\Omega_j=\caU(\rho_j)$ and the ensemble $\{\rho_j, p_j\}_j$ of test states is distributed according to the Haar measure.  Then we have
\begin{align}
\Theta&=\Theta_\rmP= \frac{d|\Phi\>\<\Phi|+1}{d+1}, \label{eq:PreparationOptimal}\\
p_\rmE(\Theta,\epsilon)&=1-\frac{d}{d+1}\epsilon, \quad p_\rmA(\Theta,\epsilon)=1-\epsilon,
\end{align} 
which implies that 
\begin{equation}\label{eq:NumTestMin}
N_\rmA(\epsilon,\delta,\Theta)=\biggl\lceil
\frac{\ln\delta}{\ln(1-\epsilon)}\biggr\rceil
\leq
\biggl\lceil
\frac{1}{\epsilon}\ln\frac{1}{\delta}\biggr\rceil
\end{equation}
and that  $N_\rmE(\epsilon,\delta,\Theta)=N_\rmA(d\epsilon/(d+1),\delta,\Theta)$ [cf.~\eref{eq:NumTestNA}].
These results are independent of the unitary transformation to be verified.
For a general verification operator $\Theta$, the probability $p_\rmA(\Theta,\epsilon)$ satisfies the bound $p_\rmA(\Theta,\epsilon)\geq 1-\epsilon$, so the number of tests required
cannot be smaller.

Recall that an ensemble of pure states $\{\rho_j, p_j\}_j$ is a 2-design \cite{ReneBSC04,RoyS07,ZhuH19O} if 
\begin{equation}
\sum_j p_j\rho_j\otimes \rho_j^*=\frac{d|\Phi\>\<\Phi|+1}{d(d+1)}.
\end{equation}
If the ensemble of test states forms a 2-design, then \eref{eq:PreparationOptimal} still holds. Therefore, in the above optimal verification protocol,  Haar random test states can be replaced by pseudo-random test states drawn from a 2-design, which is much easier to realize. These observations yield the following proposition.

\begin{proposition}
Any verification strategy characterized by the process verification operator $\Theta$ satisfies
\begin{equation}\label{eq:NumTestMin2}
p_\rmA(\Theta,\epsilon)\geq 1-\epsilon,\quad 
N_\rmA(\epsilon,\delta,\Theta)\geq \biggl\lceil
\frac{\ln\delta}{\ln(1-\epsilon)}\biggr\rceil. 
\end{equation}
Both lower bounds are saturated if the ensemble of test states forms a 2-design, and the verification operator for each output state is equal to the state projector.
\end{proposition}

\section{Balanced strategies}              
The ensemble of test states $\{\rho_j, p_j\}_j$ is \emph{balanced} if $d\sum_j p_j\rho_j=1$, so that $\{dp_j\rho_j\}_j$ is formally a positive operator-valued measure.
In this case,  the verification strategy and the  operators $\Theta,\Theta_\rmP$ are also called balanced. 
Then we have
\begin{equation}\label{eq:PVObalanced}
|\Phi\>\<\Phi|\leq \Theta_\rmP\leq \Theta\leq1, 
\end{equation}
and  the maximally entangled state $|\Phi\>$ is an eigenstate of $\Theta_\rmP$ and $\Theta$ with the largest  eigenvalue~1. Formally, $\Theta_\rmP$ and $\Theta$ are verification operators of $|\Phi\>$, so many results on the verification of maximally entangled states presented in \rcite{ZhuH19O} can be applied to the current study. 
Here is a simple way for constructing  balanced ensembles: If $\{|\psi_j\>\}_j$ forms an orthonormal basis in $\caH$, then the ensemble $\{|\psi_j\>\<\psi_j|,p_j=1/d\}_j$ is balanced. In this case, $\{U|\psi_j\>\}_j$ also forms an orthonormal basis, which is  helpful to simplifying the verification of the output states. In the rest of this paper, we only consider balanced strategies unless it is stated otherwise.

When  $\Theta_\rmP$ and $\Theta$ are balanced, the \emph{spectral gap} $\nu$   of the verification strategy is defined as the spectral gap of $\Theta$, that is,  $\nu=\nu(\Theta):=1-\beta(\Theta)$, where $\beta=\beta(\Theta)$ is the second largest eigenvalue of  $\Theta$. By contrast, the \emph{preparation spectral gap} $\nu_\rmP$  of the verification strategy is defined as the spectral gap of $\Theta_\rmP$, that is, $\nu_\rmP:=1-\beta_\rmP$, where $\beta_\rmP=\beta(\Theta_\rmP)$ is the second largest eigenvalue of  $\Theta_\rmP$. Note that $\beta$ and $\beta_\rmP$ are the operator norms of $\Theta-|\Phi\>\<\Phi|$ and $\Theta_\rmP-|\Phi\>\<\Phi|$, respectively, that is,
\begin{equation}\label{eq:betabetaP}
\beta=\|\Theta-|\Phi\>\<\Phi|\|, \quad \beta_\rmP=\|\Theta_\rmP-|\Phi\>\<\Phi|\|. 
\end{equation}
We have $\beta_\rmP\leq \beta$ and $\nu\leq \nu_\rmP$ given that $\Theta_\rmP\leq \Theta$. In addition,   \eref{eq:ThetaBound} implies that 
\begin{equation}\label{eq:nuLB}
\Theta\leq \nu_\rmM\Theta_\rmP+\beta_\rmM,
\quad \nu\geq \nu_\rmP\nu_\rmM
\end{equation}
given that  the verification strategy is balanced. 
\begin{proposition}\label{pro:pEUB} Suppose $\Theta$ is the process verification operator of a balanced strategy. Then 
\begin{align}
p_\rmE(\Theta,\epsilon)&\leq 1-\nu \epsilon\leq 1-\nu_\rmP\nu_\rmM \epsilon,\label{eq:pEUB} \\
N_\rmE(\epsilon,\delta,\Theta)&\leq \biggl\lceil
\frac{\ln\delta}{\ln(1-\nu\epsilon)}\biggr\rceil\leq  \biggl\lceil
\frac{\ln\delta}{\ln(1-\nu_\rmP\nu_\rmM\epsilon)}\biggr\rceil
\nonumber\\ &\leq \biggl\lceil
\frac{\ln(\delta^{-1})}{\nu_\rmP\nu_\rmM\epsilon}\biggr\rceil. \label{eq:NumTestBal}
\end{align} 
\end{proposition}
\begin{proof}
In view of \esref{eq:MaxPPEF} and \eqref{eq:pESDP},  the first upper bound for $p_\rmE(\Theta,\epsilon)$  in \eref{eq:pEUB}  
can be proved using a similar idea used to prove Eq.~(1) in \rcite{ZhuH19AdL} (cf.~\rcite{PallLM18}). The second upper bound follows from the inequality $\nu\geq \nu_\rmP\nu_\rmM$ in \eref{eq:nuLB}.  
\Eref{eq:NumTestBal} is a simple corollary of  \esref{eq:NumTest} and \eqref{eq:pEUB}. 
\end{proof}
According to \pref{pro:pEUB}, the efficiency of a balanced strategy is mainly determined by the preparation spectral gap  $\nu_\rmP$ and the measurement spectral gap  $\nu_\rmM$, which characterize the performances of state preparation and measurements, respectively.

\begin{proposition}
For any balanced strategy, the preparation spectral gap satisfies $\nu\leq \nu_\rmP\leq d/(d+1)$; the upper bound for $\nu_\rmP$ is saturated iff  the ensemble of test states is a 2-design. The measurement spectral gap satisfies $\nu_\rmM\leq 1$; the bound is saturated iff  the verification operator for each output state equals  the state projector.
\end{proposition}
\begin{proof}
The inequality	$\nu\leq \nu_\rmP$ follows from the fact that $\Theta\geq \Theta_\rmP$. 
The inequality $\nu_\rmP\leq d/(d+1)$ is equivalent to the inequality  $\beta_\rmP\geq 1/(d+1)$, which follows from  \eref{eq:PVObalanced}   and the equality $\tr(\Theta_\rmP)=d$.
If the inequality is saturated, then $\Theta_\rmP= (d|\Phi\>\<\Phi|+1)/(d+1)$, which implies that the ensemble of test states is a 2-design, in which case the inequality is indeed saturated. The inequality $\nu_\rmM\leq 1$ and the equality condition follow from the definition of $\nu_\rmM$ in \eref{eq:nuM}.
\end{proof}

When $d\geq 3$, it is known that a 2-design can be constructed from $\lceil\frac{3}{4}(d-1)^2\rceil+1$ bases  \cite{RoyS07,LiHZ19}. When $d$ is a prime power, the set of stabilizer states is a 2-design. In addition,
a 2-design can be constructed from a complete set of $d+1$ mutually unbiased bases (MUB) \cite{KlapR05M,RoyS07,Zhu15M}. Recall that two bases $\{|\psi_j\>\}_j$ and $\{|\varphi_k\>\}_k$ are mutually unbiased if $|\<\psi_j|\varphi_k\>|^2=1/d$ for all $j$ and $k$ \cite{Ivan81,WootF89,DurtEBZ10}. If instead the ensemble $\{\rho_j,p_j\}_j$  is constructed from $r$ MUB with uniform weights, then we have $\nu_\rmP=(r-1)/r$ (cf.~Proposition~2 in \rcite{ZhuH19O}). The case of two MUB was studied in \rcite{Hofm05} (cf.~\rcite{MayeK18}).

Define the phase  operator $Z$ and the  shift operator $X$ on $\caH$ as follows,
 \begin{equation}\label{eq:ZX}
 Z|j\>=\omega^j |j\>,\quad X |j\>=|j+1\>, \quad j\in \bbZ_d, 
 \end{equation}
 where $\omega=\rme^{2\pi\rmi/d}$ and $\bbZ_d$ is the ring of integers modulo $d$. Then the two operators generate the (generalized) Pauli group up to overall phase factors. Here we are mainly interested in the fact that
the respective eigenbases of the three operators $Z, X, XZ$  are mutually unbiased \cite{DurtEBZ10}. To be specific, the eigenbasis of $Z$ is just the standard computational basis; the eigenbasis of $X$ is  the Fourier basis and is composed of the kets 
\begin{equation}
\frac{1}{\sqrt{d}}\sum_{j\in \bbZ_d}\omega^{-sj}|j\>, \quad s\in \bbZ_d;
\end{equation}
the eigenbasis of $XZ$ is composed of the kets
\begin{equation}
\frac{1}{\sqrt{d}}\sum_{j\in \bbZ_d}\tau^{(s-j)^2}|j\>, \quad s\in \bbZ_d,
\end{equation}
where $\tau=-\rme^{\pi\rmi/d} $.    If  the dimension $d$ is a  prime, then the respective eigenbases of $Z, X, XZ, XZ^2,..., XZ^{d-1}$ constitute a complete set of  MUB \cite{DurtEBZ10}.

\section{Local state preparation and measurements}               
In practice, it is usually not easy to prepare entangled test states or to perform entangling measurements. It is thus of fundamental interest to clarify  the limitation of local operations on QGV. Suppose the Hilbert space $\caH$ is a tensor product of $n$ Hilbert spaces $\caH=\bigotimes_{k=1}^n \caH_k$ with $\dim (\caH_k)=d_k$ and $d=\prod_{k=1}^n d_k$.  Let $\rmU(\caH_k)$ be the group of unitary operators on $\caH_k$. Suppose we can only prepare product test states and only consider balanced strategies. Then the preparation spectral gap $\nu_\rmP$ is maximized when the product test states are distributed according to the normalized Haar measure induced by $\bigotimes_{k=1}^n\rmU(\caH_k)$ given that $\nu_\rmP=1-\beta_\rmP$ is concave in $\Theta_\rmP$ thanks to \eref{eq:betabetaP}. 
In this case, we have
\begin{equation}
\Theta_\rmP=\bigotimes_{k=1}^n\frac{d_k|\Phi_k\>\<\Phi_k|+1}{d_k+1}, \quad \nu_\rmP=\frac{d_{\min}}{d_{\min}+1},
\end{equation}
where $|\Phi_k\>$ is a maximally entangled state in $\caH_k\otimes \caH_k$ and $d_{\min}=\min_k d_k$ is the minimum of the local dimensions. Note that $\Theta_\rmP$ acts on $(\bigotimes_{k=1}^n \caH_k)^{\otimes 2}$, which is isomorphic to  $\bigotimes_{k=1}^n \caH_k^{\otimes 2}$. 
The above equation also applies if the ensemble of test states is constructed from a product  2-design, that is,  the tensor product of 2-designs for  individual local Hilbert spaces. These conclusions are summarized in the following proposition.
\begin{proposition}
Suppose the ensemble of test states is balanced and only consists of product states on the Hilbert space $\caH$. Then the preparation spectral gap satisfies  $\nu\leq \nu_\rmP\leq d_{\min}/(d_{\min}+1)$. The upper bound for $\nu_\rmP$ is saturated if the ensemble forms a product 2-design. 
\end{proposition}

Suppose we can construct $r$ MUB  $\caB_1^{(k)}, \caB_2^{(k)}, \ldots, \caB_r^{(k)}$ for each $\caH_k$. Then the $r$ bases $\bigotimes_{k=1}^n \caB_s^{(k)} $ in $\caH$ for $s=1,2,\ldots, r$ are mutually unbiased. These bases can be used to construct an ensemble of test states, which can achieve the preparation spectral gap  $\nu_\rmP=(r-1)/r$. If each local dimension $d_k$ for $k=1,2,\ldots,n$ is a prime power, then we can construct $d_{\min}$ mutually unbiased product bases in this way, so the preparation spectral gap can attain the maximum $d_{\min}/(d_{\min}+1)$ (over local preparation strategies).  In general, we can construct at least three mutually unbiased product bases using eigenstates of Pauli operators given the discussion after \eref{eq:ZX}, so the preparation spectral gap achievable in this way is at least $2/3$. The restriction to  local state preparation thus has little  influence  on the verification efficiency.

By contrast, the limitations of local measurements depend on the unitary transformation to  be verified and the test states chosen. Nevertheless,  many important quantum 
gates and circuits can be verified efficiently
using local measurements in addition to 
local state preparation, as shown in the next section.

\section{Applications}               
To illustrate the power of the general framework of QGV proposed above, here we construct efficient balanced verification protocols for a number of important quantum gates and circuits using local operations.

Let us start with bipartite unitary transformations. The test states can be constructed from a product  2-design, so that the preparation spectral gap can attain the maximum $d_{\min}/(d_{\min}+1)$ (over local state preparation). Alternatively, we can achieve a preparation spectral gap of $2/3$ by virtue of  mutually unbiased product bases
constructed from eigenstates of Pauli operators. To verify general bipartite unitary transformations efficiently with local measurements, we need to verify general bipartite pure states efficiently with local measurements. Fortunately, the later problem has been solved recently \cite{PallLM18,LiHZ19,YuSG19,WangH19}. For any bipartite pure state, we can construct a verification strategy using local measurements which can achieve a spectral  gap of at least $2/3$. Such an efficient strategy can be constructed by performing the standard test (based on projective measurements on the Schmidt basis) and a number of adaptive local projective tests with  two-way communication. In addition, the measurement spectral gap can reach the value of  $1/2$ even if we can perform only two distinct tests
based on one-way communication \cite{LiHZ19}.  The number of tests required by such a simplified strategy is comparable to the optimal strategy. Therefore, all bipartite unitary transformations can be verified efficiently with local state preparation and measurements as summarized in \tref{tab:NumTestQGV}. By contrast, RB usually only applies to gate sets with special structure, such as the Clifford group \cite{EmerAZ05,KnilLRB08,MageGE11,WallF14}. Although generalizations have been considered by many researchers \cite{MageGJR12,CariWE15,CrosMBS16,HarpF17,OnorWE19,HelsXVW19}, a simple recipe for a generic bipartite unitary transformation is still not available.

\begin{table*}
	\caption{\label{tab:NumTestQGV}
		Verification of bipartite and multipartite unitary transformations based on local state preparation and local projective measurements. All strategies considered are balanced. Here $\nu_\rmP$ and $\nu_\rmM$ denote the preparation spectral gap and the measurement spectral gap; $\nu$  (LB) denotes a lower bound for the spectral gap. $N_\rmE(\epsilon,\delta)$ (UB) denotes an upper bound for the minimum number of tests required to verify the target unitary transformation within entanglement infidelity $\epsilon$ and significance level $\delta$.
	}		
	\begin{math}
	\begin{array}{c|ccccccc}
	\hline\hline
	\mbox{Unitary} &\mbox{Local dimension}&\mbox{Preparation} &\mbox{Measurement}& \nu_\rmP &\nu_\rmM &\nu \mbox{ (LB)} & N_\rmE(\epsilon, \delta) \mbox{ (UB)}\\[0.5ex]
	\hline
	\mbox{Bipartite} & d_1, d_2&  \mbox{Pauli}  &\mbox{Local adaptive} & \frac{2}{3}&\frac{2}{3} & \frac{4}{9}& \bigl\lceil\frac{9}{4} \epsilon^{-1}\ln\delta^{-1}\bigr\rceil \\[0.5ex]
	\mbox{Bipartite} & d_1, d_2& \mbox{Product}&\mbox{Local adaptive}& \frac{d_{\min}}{d_{\min}+1}&\frac{2}{3} & \frac{2d_{\min}}{3(d_{\min}+1)} & \bigl\lceil\frac{3(d_{\min}+1)}{2d_{\min}} \epsilon^{-1}\ln\delta^{-1}\bigr\rceil \\[0.5ex]
	\mbox{Clifford}&2 &\mbox{Pauli} & \mbox{Pauli} & \frac{2}{3} & \frac{1}{2}&\frac{1}{3} & \lceil 3\epsilon^{-1}\ln\delta^{-1}\rceil \\[0.5ex]
	\mbox{Clifford}& d_1 \mbox{ odd prime}& \mbox{Pauli} & \mbox{Pauli} &  \frac{d_1}{d_1+1} &\frac{d_1-1}{d_1}& \frac{d_1-1}{d_1+1}&	
	\bigl\lceil \frac{d_1+1}{d_1-1}\epsilon^{-1}\ln\delta^{-1}\bigr\rceil \\[0.5ex]		
	\mbox{Clifford}& d_1 & \mbox{Pauli} & \mbox{Pauli} &\frac{2}{3} &\frac{1}{n}&	\frac{2}{3n}&
	\bigl\lceil \frac{3n}{2}\epsilon^{-1}\ln\delta^{-1}\bigr\rceil \\[0.5ex]		
	C^{(n-1)}Z  & d_1 & \mbox{Pauli} & \mbox{Pauli} & \frac{1}{2}& \frac{1}{n}&\frac{1}{2n}& \lceil 2n\epsilon^{-1}\ln\delta^{-1}\rceil\\[0.5ex]
	C^{(n-1)}X  &d_1 & \mbox{Pauli} & \mbox{Pauli} & \frac{1}{2}& \frac{1}{n}&\frac{1}{2n}& \lceil 2n\epsilon^{-1}\ln\delta^{-1}\rceil\\[0.5ex]
	\mbox{CSWAP} &2 & \mbox{Pauli} & \mbox{Pauli} & \frac{2}{3}& \frac{2}{3}&\frac{4}{9}& \bigl\lceil \frac{9}{4}\epsilon^{-1}\ln\delta^{-1}\bigr\rceil\\[0.5ex]
	\mbox{Permutation}  &d_1 & \mbox{Pauli} & \mbox{Pauli} & \frac{2}{3}& 1 &\frac{2}{3}& \bigl\lceil \frac{3}{2}\epsilon^{-1}\ln\delta^{-1}\bigr\rceil\\[0.5ex]
	\mbox{Permutation}  &d_1 & \mbox{Product} & \mbox{Product} & \frac{d_1}{d_1+1}& 1 & \frac{d_1}{d_1+1}& \bigl\lceil \frac{d_1+1}{d_1}\epsilon^{-1}\ln\delta^{-1}\bigr\rceil\\[0.5ex]
	\hline\hline
	\end{array}	
	\end{math}
\end{table*}

Next, consider the verification of $n$-qubit  Clifford unitaries. Here we only require Pauli operations: preparation of eigenstates of local Pauli operators and measurements of local Pauli operators. 
The ensemble of test states can be constructed from three mutually unbiased product bases, namely, the standard basis, the eigenbasis of all $X$ operators for individual qubits, and the eigenbasis of all $Y=\rmi XZ$ operators. In this way, the preparation spectral gap can attain the maximum  of $2/3$.  Under the action of a Clifford unitary, each test state is turned into a stabilizer state, which can be verified efficiently with Pauli measurements. More precisely,  we can measure all nontrivial stabilizer operators with an equal probability, and the resulting measurement spectral gap is $2^{n-1}/(2^n-1)\geq 1/2$ \cite{PallLM18,ZhuH19AdL} (the performance can be improved if the stabilizer state is a product state). These results apply to both individual gates and gate sets, including  whole Clifford circuits.

When  the local dimension $d_1$ is an odd prime, the ensemble of test states can be constructed using $d_1+1$ mutually unbiased product  bases  from eigenstates of Pauli operators, in which case the preparation spectral gap is $d_1/(d_1+1)$.
In addition, each output stabilizer state can be verified  efficiently using  a protocol proposed in  \rcite{ZhuH19AdL}, and the resulting  measurement spectral gap is
\begin{equation} \nu_\rmM=\frac{d_1^{n}-d_1^{n-1}}{d_1^n-1}\geq \frac{d_1-1}{d_1}.
\end{equation}
When  the local dimension is not a prime, the ensemble of test states can be constructed using three mutually unbiased product  bases  from eigenstates of Pauli operators, in which case the preparation spectral gap is  $2/3$. Each output stabilizer state can be verified by measuring  $n$ stabilizer generators with an equal probability; the resulting measurement spectral gap  is $1/n$. The performance may be improved by virtue of the cover or coloring protocol proposed in \rcite{ZhuH19E}. Therefore, Clifford gates and circuits can be verified efficiently even if the local dimension is not a prime (see \tref{tab:NumTestQGV}). With respect to the infidelity, the number of tests required by our protocol is quadratically 
fewer  than what is required by direct fidelity estimation \cite{FlamL11}.

Next, consider the verification of the generalized controlled-$Z$ gate $C^{(n-1)}Z$ with $n-1$ control qubits. Now test states can be constructed from two mutually unbiased product bases, namely, the computational basis
and the eigenbasis of all   $X$ operators for individual qubits. The resulting preparation spectral gap is $1/2$.  Under the action of $C^{(n-1)}Z$, each state in the computational basis is invariant and so can be verified by performing the projective measurement on the computational basis. By contrast, each state in the other basis is turned into an order-$n$ hypergraph state \cite{QuWLB13,RossHBM13}, which can be verified efficiently using Pauli measurements thanks to the coloring protocol proposed in \rcite{ZhuH19E};  the resulting measurement spectral gap is $1/n$.
The same method can also be applied to verifying
circuits composed of generalized controlled-$Z$ gates, which is useful to studying instantaneous quantum polytime (IQP) circuits and to demonstrating quantum supremacy \cite{ShepB09, BremMS16,ArutABB19ab}. 
With minor modification, this approach can be applied to verifying  generalized controlled-$X$  gates (CNOT), including the Toffoli gate \cite{NielC10book}, given that $C^{(n-1)}X= H_n C^{(n-1)}ZH_n$, where $H_n$ is the Hadamard gate acting on the $n$th qubit. Generalization to the qudit setting is also straightforward.

Next, consider  the   CSWAP gate. The test states can be constructed from three mutually unbiased product bases as in the verification of Clifford unitaries. Under the action of the CSWAP gate, each state in the computational basis  still belongs to the computational basis and so can easily be verified. Each state in the other two bases is either invariant (still  a product state) or turned into a Greenberger-Horne-Zeilinger (GHZ) state up to a local Clifford transformation; the former can be verified by a Pauli measurement, while the latter can be verified by an efficient protocol based on Pauli measurements \cite{LiHZ19O}. The resulting measurement spectral gap is $2/3$. So the CSWAP gate can be verified efficiently as well.

Finally, let us consider the verification of permutation transformations of a multiqudit system with local dimension $d_1$. Compared with generic Clifford unitary transformations, these permutations can be verified slightly more efficiently. 
The ensemble of test states can be constructed from a product 2-design, so that the preparation spectral gap can attain the maximum $d_1/(d_1+1)$. Meanwhile, the measurement spectral gap can attain the maximum 1 since any permutation of a product state is still a product state. Therefore, the spectral gap 
of the process verification operator can attain the maximum  $d_1/(d_1+1)$ over all balanced strategies based on local operations.
Alternatively, the ensemble of test states can be constructed from $r$ mutually unbiased product bases (with $2\leq r\leq d_1+1$), in which case the preparation spectral gap is $\nu_\rmP=(r-1)/r$. Such a strategy is much easier to implement, and the number of tests required is at most two times compared with the optimal strategy even if $r=2$. For example, three mutually unbiased product bases can be constructed from eigenstates of Pauli operators,  which yield $\nu_\rmP=2/3$.  When $d_1$ is a prime,  $d_1+1$ mutually unbiased product bases can be constructed from eigenstates of Pauli operators, which yield $\nu_\rmP=d_1/(d_1+1)$.

\section{Summary}               
We proposed  a simple and general framework for verifying unitary transformations that can be applied to both individual gates and gate sets, including quantum circuits. Our approach enables   efficient and scalable  verification of all bipartite unitaries, Clifford unitaries, generalized controlled-$Z$ gates,  generalized CNOT gates, the CSWAP gate, and permutation transformations using local state preparation and measurements. Moreover, only Pauli operations are required in many cases of practical interest as shown in \tref{tab:NumTestQGV}. It is applicable in a number of scenarios in which RB does not apply and is thus instrumental to quantum computation and various other applications in quantum information processing. Nevertheless, our approach  is not a replacement for RB since the latter is more robust against SPAM. Instead, combining the merits of our approach and those of  RB may lead to a more powerful tool. We hope that our work can stimulate further progresses in this direction.

%

\acknowledgments
This work is  supported by  the National Natural Science Foundation of China (Grant No. 11875110).

\emph{Note added}. Recently, we became aware of  related works by Liu et al.  \cite{LiuSYZ20} and Zeng et al. \cite{ZengZL19}.



\nocite{apsrev41Control}
\bibliographystyle{apsrev4-1}
\bibliography{all_references}

\end{document}